\documentclass[a4paper,10pt]{article}
\usepackage[mathlines]{lineno}
\usepackage[bbgreekl]{mathbbol}
\usepackage{soul}
\usepackage[small]{titlesec}
\usepackage[utf8]{inputenc}
\usepackage[T1]{fontenc}  
\usepackage[]{mdframed}
\usepackage{mathrsfs}
\usepackage{caption}
\usepackage[noerroretextools,backend=biber,style=alphabetic,natbib=false,giveninits=true,doi=true,isbn=false,url=false,date=year,maxbibnames=99,sorting=nty,defernumbers=true]{biblatex}
\DeclareNameAlias{default}{family-given}

\DeclareFieldFormat[article]{title}{#1}
\renewbibmacro{in:}{}
\DeclareFieldFormat[article]{pages}{#1}
\DeclareFieldFormat{journal}{#1}
\renewcommand\bf\bfseries

\renewbibmacro*{volume+number+eid}{%
	\printfield{volume}%
	\setunit*{\addnbspace}% NEW (optional); there's also \addnbthinspace
	\printfield{number}%
	\setunit{\addcomma\space}%
	\printfield{eid}}
\DeclareFieldFormat[article]{number}{\mkbibparens{#1}}
\DeclareFieldFormat[article]{volume}{\textbf{#1}}
\DeclareFieldFormat{year}{\mkbibparens{#1}}
\DeclareBibliographyDriver{article}{%
	\printnames{author}:%
	\newunit\newblock
	\printfield{title}%
	\newunit\newblock
	\printfield{journaltitle}
	\newunit
	\iffieldundef{number}{\printfield{volume}}{\printfield{volume}\addspace\printfield{number}}
	\addcomma\addspace\printfield{pages}\addspace
	\printfield{year}
}
\AtEveryBibitem{\clearfield{month}}
\AtEveryBibitem{\clearfield{day}}
\usepackage[english]{babel}
\usepackage[T1]{fontenc}
\usepackage{csquotes}
\usepackage{bbm}
\usepackage[leqno]{amsmath}
\usepackage{amsfonts,amsthm,amsbsy,amssymb,dsfont,stmaryrd}
\usepackage{aliascnt}
\usepackage[dvipsnames]{xcolor}
\usepackage{braket}

\usepackage{caption}
\usepackage{subcaption}

\usepackage{enumitem}
\usepackage{environ}

\makeatletter
\newcommand{\leqnomode}{\tagsleft@true\let\veqno\@@leqno}
\newcommand{\reqnomode}{\tagsleft@false\let\veqno\@@eqno}
\makeatother

\usepackage{mathtools}
\usepackage[makeroom]{cancel}
\numberwithin{equation}{section}

\newcommand\myshade{85}
\colorlet{mylinkcolor}{violet}
\colorlet{mycitecolor}{YellowOrange}
\colorlet{myurlcolor}{Aquamarine}

\usepackage[left=2.5cm,right=2.5cm,top=2.5cm,bottom=2.5cm]{geometry}
\usepackage[unicode=true,pdfusetitle,
bookmarks=true,bookmarksnumbered=false,bookmarksopen=false,
breaklinks=false,pdfborder={0 0 1},backref=false,
linkcolor  = mylinkcolor!\myshade!black,
citecolor  = mycitecolor!\myshade!black,
urlcolor   = myurlcolor!\myshade!black,
colorlinks = true,
]
{hyperref}

\usepackage{ifthen}

\usepackage{graphicx}
\usepackage{caption}
\usepackage{slashed}
\definecolor{ct_black}{HTML}{000000}
\definecolor{ct_orange}{HTML}{ED872D}
\definecolor{ct_purple}{HTML}{7A68A6}
\definecolor{ct_blue}{HTML}{348ABD}
\definecolor{ct_turquoise}{HTML}{188487}
\definecolor{ct_red}{HTML}{E32636}
\definecolor{ct_pink}{HTML}{CF4457}
\definecolor{ct_green}{HTML}{467821}

\definecolor{ct2_green}{HTML}{9FF781}
\definecolor{ct2_green_dark}{HTML}{088A08}

\theoremstyle{plain}
\newtheorem{thm}{\protect\theoremname}[section]

\newaliascnt{lem}{thm}
\theoremstyle{plain}
\newtheorem{lem}[lem]{\protect\lemmaname}
\aliascntresetthe{lem}

\newaliascnt{cor}{thm}
\theoremstyle{plain}

\aliascntresetthe{cor}

\newaliascnt{prop}{thm}
\theoremstyle{plain}
\newtheorem{prop}[prop]{\protect\propositionname}
\aliascntresetthe{prop}

\newaliascnt{claim}{thm}
\theoremstyle{plain}
\newtheorem{claim}[claim]{\protect\claimname}
\aliascntresetthe{claim}

\newaliascnt{assumption}{thm}

\aliascntresetthe{assumption}

\theoremstyle{remark}
\newaliascnt{rem}{thm}

\aliascntresetthe{rem}

\theoremstyle{definition}
\newaliascnt{defn}{thm}

\aliascntresetthe{defn}

\newaliascnt{example}{thm}

\aliascntresetthe{example}
\providecommand{\assumptionname}{Assumption}
\providecommand{\claimname}{Claim}
\providecommand{\corollaryname}{Corollary}
\providecommand{\definitionname}{Definition}
\providecommand{\lemmaname}{Lemma}
\providecommand{\propositionname}{Proposition}
\providecommand{\remarkname}{Remark}
\providecommand{\theoremname}{Theorem}
\providecommand{\examplename}{Example}

\usepackage[nameinlink]{cleveref}

\crefname{section}{Section}{Sections}
\crefname{example}{Example}{Examples}
\crefname{appendix}{Appendix}{Appendices}
\crefname{figure}{Figure}{Figures}
\crefname{assumption}{Assumption}{Assumptions}
\crefname{thm}{Theorem}{Theorems}
\crefname{lem}{Lemma}{Lemmas}
\crefname{cor}{Corollary}{Corollaries}
\crefname{prop}{Proposition}{Propositions}
\crefname{claim}{Claim}{Claims}
\crefname{rem}{Remark}{Remarks}
\crefname{defn}{Definition}{Definitions}
\Crefname{assumption}{Assumption}{Assumptions}
\Crefname{thm}{Theorem}{Theorems}
\Crefname{lem}{Lemma}{Lemmas}
\Crefname{cor}{Corollary}{Corollaries}
\Crefname{prop}{Proposition}{Propositions}
\Crefname{claim}{Claim}{Claims}
\Crefname{rem}{Remark}{Remarks}
\Crefname{defn}{Definition}{Definitions}
\Crefname{example}{Example}{Examples}
\crefformat{equation}{(#2#1#3)}
\crefname{table}{Table}{Tables}

\crefrangelabelformat{equation}{(#3#1#4--#5#2#6)}

\crefmultiformat{equation}{(#2#1#3}{, #2#1#3)}{#2#1#3}{#2#1#3}
\Crefmultiformat{equation}{(#2#1#3}{, #2#1#3)}{#2#1#3}{#2#1#3}

\crefname{appendix}{Appendix}{Appendices}
\Crefname{appendix}{Appendix}{Appendices}

\newtheorem*{lem*}{\protect\lemmaname}

\newcommand{\ee}{\operatorname{e}}
\newcommand{\ii}{\operatorname{i}}

\newcommand{\RR}{\mathbb{R}}

\newcommand{\CC}{\mathbb{C}}

\newcommand\norm[1]{\left\lVert#1\right\rVert}
\newcommand\abs[1]{\left|#1\right|}
\newcommand{\ip}[2]{\langle #1, #2 \rangle}

\newcommand{\dif}{\operatorname{d}\!} % shorten the space afterwards

\renewcommand{\Re}[1]{\operatorname{\mathbb{R}\mathbbm{e}}\{#1\}}
\newcommand{\Arg}[1]{\operatorname{\operatorname{Arg}}\left(#1\right)}

\newcommand{\Id}{\mathds{1}}
\newcommand{\HH}{\mathbb{H}}

\NewEnviron{malign}{%
	\begin{align}\begin{split}
			\BODY
	\end{split}\end{align}
}

\newcommand{\eq}[1]{\begin{align*}#1\end{align*}}

\newcommand{\eql}[1]{\begin{align}#1\end{align}}

\newcommand{\br}[1]{\left(#1\right)}

\title{Average lower bound on magnetic double-well tunneling with asymptotically large well separation}
\author{
    \href{mailto:cf@math.princeton.edu}{Charles L. Fefferman}\\
    {\footnotesize Department of Mathematics, Princeton University}\\    \href{mailto:sk9017@princeton.edu}{Say-Yeon Kwon}\\
{\footnotesize Department of Mathematics, Princeton University}\\
    \href{mailto:jacobshapiro@princeton.edu}{Jacob Shapiro}\\
{\footnotesize Department of Mathematics, Princeton University}\\
{\footnotesize and Faculty of Mathematics, University of Vienna}\\    \href{mailto:miw2103@columbia.edu}{Michael I. Weinstein}\\
    {\footnotesize Department of Applied Physics and Applied Mathematics,}\\
    {\footnotesize and Department of Mathematics, Columbia University}
}
\date{}
\usepackage{tikz}
\begin{document}

\maketitle
%\linenumbers
\begin{abstract}
We study tunneling between two identical, smooth, compactly supported
potential wells in a constant perpendicular magnetic field, with the
coupling and magnetic field fixed and the well separation tending to
infinity. For a sufficiently small positive magnetic field, we prove
lower bounds on the hopping coefficient and on the splitting of the two
lowest eigenvalues outside a set of separations of density zero. If the
well centers are separated by $2d_1$ and the magnetic field strength is
$b\lambda$, the lower bounds have the form
$\ee^{-b\lambda d_1^2-Cd_1^{1+\delta}}$ for every fixed $\delta\in\br{0,1}$.
No radial or reflection symmetry of the single-well potential is assumed.
The proof combines the explicit Landau resolvent kernel with a complex
analytic estimate for logarithmic averages of the hopping coefficient.
\end{abstract}

\section{Introduction}\label{sec:intro}

Quantum tunneling couples states localized near distinct potential wells.
For two identical, sufficiently deep or sufficiently distant wells, the
single-well ground-state energy gives rise to a pair of nearby energy
levels. Their splitting determines the time scale for transfer between
the wells: a smaller splitting corresponds to a longer tunneling time.
In the absence of a magnetic field, the semiclassical analysis of this
splitting is classical. The works of Simon and Helffer--Sj\"ostrand
\cite{Simon1984,HelfferSjostrand1984} relate its exponential scale to the
geometry of the intervening classically forbidden region. Positivity
of the nonmagnetic ground state is an essential ingredient in obtaining
lower bounds. Related lower bounds for compactly supported wells enter
the strong-binding analysis of honeycomb Schr\"odinger operators
\cite{FLW2018}.

A magnetic field changes the picture. Magnetic
translations \cite{Zak1964} multiply translated bound states by complex
phases, and the interaction between the wells is described by an
oscillatory hopping integral. Decay estimates give upper bounds on its
absolute value, but do not exclude cancellations. Helffer and Sj\"ostrand obtained precise semiclassical
tunneling asymptotics for sufficiently small magnetic coupling, under
geometric and analyticity assumptions near the wells and the relevant
minimizing Agmon geodesics \cite{HelfferSjostrand1987}. Their cosine
formula in the symmetric two-geodesic setting yields, in particular,
sharp lower bounds whenever the oscillatory factor is bounded away
from zero \cite[Section~4, (4.13)--(4.14)]{HelfferSjostrand1987}.
Here we study lower bounds for smooth, compactly supported wells in
the large-separation limit.

The work of Fefferman, Shapiro and Weinstein on this question began with
\cite{SIAM2022_FSW}, published in 2022. For a class of radial single-well potentials,
they proved upper and lower bounds on ground-state tunneling in the
simultaneous limit of deep wells and strong magnetic fields: the
potential is scaled by $\lambda^2$, the field strength is proportional
to $\lambda$, and $\lambda\to\infty$ at fixed well separation. Radial
symmetry permits a lower bound on the oscillatory hopping coefficient
and, through a reduction to the two localized states, on the eigenvalue
splitting. These estimates also enter the derivation of magnetic
tight-binding models by Shapiro and Weinstein \cite{SW2022}.

The choice of compactly supported atomic wells is of
interest in the modeling of macroscopic crystalline quantum materials.
For example, a natural model replaces the double-well potential by a
potential formed by centering the atomic potential at each vertex of an
infinite discrete set, which can be translation invariant (periodic) or
disordered. Under a suitable non-clustering assumption on these points
and appropriate lower bounds on hopping, the large-$\lambda$ regime is
described by a discrete effective (tight-binding) magnetic Hamiltonian.
Such models are widely used in the physics community. A central point
in a rigorous derivation of a limiting Hamiltonian is an appropriate
lower bound on the double-well magnetic hopping coefficient between
minimally separated atoms \cite{SW2022}. The examples of vanishing hopping
in \cite{FSW_25_absence,FSW2025PNAS} illustrate the non-applicability of
the standard limiting tight-binding model for certain atomic well
configurations.

In \cite{FSW2025}, Fefferman, Shapiro and Weinstein studied tunneling
associated with excited single-well states in the absence of a magnetic
field. Here the asymptotic parameter is already the separation: the
coupling is fixed and the wells are moved apart. Under reflection
symmetry, a boundary representation of the hopping coefficient has a
definite sign despite the nodes of the excited state, and unique
continuation yields a lower bound. For simple isolated single-well
levels, they also establish the asymptotic relation between hopping and
eigenvalue splitting. Thus the large-separation regime considered here
has a nonmagnetic predecessor in \cite{FSW2025}.

The subsequent preprints \cite{FSW_25_absence,FSW_25_lower_bound} develop
the magnetic results announced in \cite{FSW2025PNAS}. The counterexample
in \cite{FSW_25_absence} shows that suitable smooth, compactly supported, nonpositive,
nonradial wells in a constant magnetic field can produce
an exactly degenerate double-well ground state, and hence vanishing
low-energy tunneling. This construction uses a family of potentials
whose shape may depend on the large coupling $\lambda$; it also permits
a change in the parity of the ground state. It establishes that
magnetic cancellations can eliminate the splitting, rather than merely
make it smaller.

For a fixed compactly supported potential with a unique nondegenerate
minimum, \cite{FSW_25_lower_bound} bounds hopping and splitting below
by $\ee^{-\lambda^{1+\varepsilon}}$ outside a set of couplings of density
zero as $\lambda\to\infty$. The separation is fixed and sufficiently
large; $\varepsilon>0$ is arbitrary and $b>0$ is sufficiently small.
The proof uses analytic continuation in $\lambda$ and logarithmic
averaging.

In the present paper, we fix the coupling $\lambda$ and the magnetic
field strength $B=b\lambda>0$, and let the distance $2d_1$ between the
well centers tend to infinity. Physically, increasing $\lambda$ and
increasing $d_1$ provide two ways of weakening the interaction between
the localized states. In the first regime, localization becomes
stronger while the centers remain fixed; in the second, the localized
profiles remain fixed while their centers recede. In either case the
overlap becomes small and the tunneling time becomes large. This is a
comparison of decoupling mechanisms, not an identification of the two
limits. In particular, the present result resolves the Gaussian
distance scale $\ee^{-Bd_1^2}$, with a subquadratic correction to its
exponent, for separations outside a set of density zero.
The earlier lower bound requires $\lambda^{1+\varepsilon}$ in the
exponent and does not establish $\ee^{-C\lambda}$. Here the leading
exponent is exactly $Bd_1^2$: the correction $Cd_1^{1+\delta}$, with
$\delta<1$, leaves both the quadratic power and its coefficient
unchanged.

\subsection{Framework and main results}

Fix $\lambda>0$ and a real, nonzero, nonpositive, smooth potential $v$
supported in $B_a\br{0}\subset\RR^2$. Write $d=d_1e_1$, with $d_1>a$,
and consider the single-well and double-well Hamiltonians
\eq{
 h&=\br{P-\frac12b\lambda X^\perp}^2+\lambda^2v,\\
 H&=\br{P-\frac12b\lambda X^\perp}^2
       +\lambda^2\br{v_L+v_R},
}
on $L^2\br{\RR^2}$, where $P=-\ii\nabla$,
$X^\perp=\br{-X_2,X_1}$,
$v_L\br{x}=v\br{x+d}$ and $v_R\br{x}=v\br{x-d}$.
Thus $d_1$ is half the distance between the well centers. We take
$\lambda$ such that the nonmagnetic single well has a negative
ground-state energy. For sufficiently small $b_\star>0$, we consider
$b\in\br{0,b_\star}$ and write $e_0<0$ and $\varphi$ for the simple
ground-state eigenvalue and normalized eigenfunction of $h$. We choose
the phase of $\varphi$ continuously at $b=0$.

For $u\in\RR^2$, let
\eq{
 \br{\hat R^u f}\br{x}
 &=\ee^{-\ii\frac{b\lambda}{2}x\wedge u}f\br{x-u},
 \qquad x\wedge u=x_1u_2-x_2u_1,\\
 \varphi_L&=\hat R^{-d}\varphi,
 \qquad \varphi_R=\hat R^d\varphi.
}
Following \cite{SIAM2022_FSW,FSW2025,FSW_25_lower_bound}, we define the
hopping coefficient by
\eq{
 \rho\br{d_1}=\lambda^2
 \int_{\RR^2}\overline{v_L\br{x}\varphi_L\br{x}}
                  \varphi_R\br{x}\,\dif x.
}
Thus $\rho\br{d_1}$ is the off-diagonal interaction matrix element:
$\rho\br{d_1}=\ip{\varphi_L}{\br{H-e_0\Id}\varphi_R}$.
We denote by $\Delta\br{d_1}$ the difference between the two lowest
eigenvalues of $H$. When the dependence on the magnetic parameter is
relevant, we write $\rho\br{d_1,b}$.

Throughout the large-distance argument, $b\in\br{0,b_\star}$ and
$\lambda$ are fixed; constants and distance thresholds may depend on
$b,\lambda,v$. A measurable set $Z\subset[D,\infty)$ has density zero
if $\abs{Z\cap[D,R]}/R\to0$ as $R\to\infty$, where $\abs{\cdot}$ denotes
Lebesgue measure. Our main results are the following.

\begin{thm}\label{thm:main_hopping}
   There exists $b_\star>0$ such that for each fixed $b\in\br{0,b_\star}$
   and each fixed $\delta\in\br{0,1}$,
   there are $D,C>0$ and a set $Z_\delta\subset[D,\infty)$ of density zero
   such that for all $d_1\in[D,\infty)\setminus Z_\delta$,
   \eql{
        \abs{\rho\br{d_1}} \geq \ee^{-b\lambda d_1^2-Cd_1^{1+\delta}}.
   }
\end{thm}

As a consequence of \cref{thm:main_hopping}, we obtain the corresponding
lower bound on the eigenvalue splitting.

\begin{thm}\label{thm:main_eigenvalue_splitting}
      There exists $b_\star>0$ such that for each fixed $b\in\br{0,b_\star}$
      and each fixed $\delta\in\br{0,1}$,
      there are $D,C>0$ and a set $Z_\delta\subset[D,\infty)$ of density zero
      such that for all $d_1\in[D,\infty)\setminus Z_\delta$,
      \eql{
        \Delta\br{d_1} \geq \ee^{-b\lambda d_1^2-Cd_1^{1+\delta}}.
      }
\end{thm}

The meaning of generic separation here is precisely the exclusion of a
set of density zero. The estimates do not assert a positive lower bound
at every sufficiently large separation. No radial or reflection
symmetry, or nondegeneracy of the minimum of $v$, is required. The
smallness of $b$ is used to retain a nonzero hopping coefficient at one
fixed separation by continuity from the nonmagnetic problem; all
subsequent distance estimates keep this positive $b$ fixed.
Together with the upper bound in \cref{lem:global_ub} and the
two-state reduction below, the lower bounds imply
\eq{
 \lim_{\substack{d_1\to\infty\\d_1\notin Z_\delta}}
 \frac{-\log\abs{\rho\br{d_1}}}{d_1^2}
 =\lim_{\substack{d_1\to\infty\\d_1\notin Z_\delta}}
 \frac{-\log\Delta\br{d_1}}{d_1^2}
 =b\lambda.
}

\subsection{Analytic continuation in the separation}

The principal simplification compared to \cite{FSW_25_lower_bound} is that the single-well problem does not
depend on $d_1$. In \cite{FSW_25_lower_bound}, complexifying $\lambda$
also complexifies the magnetic field and changes the single-well
eigenvalue and eigenfunction. The full Landau resolvent does not admit
the required continuation as a bounded operator on $L^2$. That work
instead constructs an analytic continuation of the resolvent acting on
compactly supported data, using a cutoff Hamiltonian and magnetic
harmonic-oscillator estimates.

Here $B$, $e_0$ and $\varphi$ remain fixed. The Landau resolvent formula
\cite[Lemma 5.1]{Korotyaev2003ATF} rewrites the hopping coefficient as
the integral in \cref{eq:def_hopping_coef}, whose integration variables
lie in the fixed compact support of $v$. The complex parameter $z\in\CC$, obtained by allowing $d_1$ to take
complex values, enters only through explicit exponential factors
and the spatial argument of Tricomi's function $U$.
For $\Re{z}>a$, that argument avoids its branch cut. For $\Re{z}>a$, that argument avoids its branch
cut. Consequently, we can continue the scalar function $\rho\br{z}$
without analytically continuing either the potential or the
single-well eigenfunction. The required continuation and sector bounds
for $U$ are recalled in \cref{appendix:tricomy}; see also
\cite{Daalhuis2025DLMF13}.

It is useful to retain the Gaussian decay exactly and study
\eq{
 B=b\lambda,\qquad q\br{z}=\ee^{Bz^2}\rho\br{z}.
}
On a translated sector with opening arbitrarily close to $\pi$,
\cref{lem:global_ub} bounds $\abs{q\br{z}}$ by $M\ee^{c\Re{z}}$.
A nonzero value at one real separation, supplied by
\cref{prop:lowerbound_single_point}, then allows us to apply
\cref{lem:complex_anal}. This gives a logarithmic average estimate
with a remainder of order $d_1^{1+\eta}$ for every fixed $\eta>0$,
after choosing the sector sufficiently wide. The resulting estimate
on the measure of the set where $\abs{\rho}$ is too small proves
\cref{thm:main_hopping}. Finally, the reduction to the two localized
states gives
\eq{
 \Delta\br{d_1}
 =2\abs{\rho\br{d_1}}\br{1+o\br{1}}
}
as $d_1\to\infty$ outside the exceptional set, which yields
\cref{thm:main_eigenvalue_splitting}.

\subsection*{Outline}

The proof of \cref{thm:main_hopping} begins with the averaged estimate
in \cref{thm:main_avg_lower_bound_hopping}.
In \cref{subsec:holomorphicity_of_rho} we obtain the integral
representation and establish holomorphicity in the separation.
In \cref{sec:global_upper_bound_rho} we bound the Gaussian quotient.
After establishing the lower bound at one point, we apply
\cref{lem:complex_anal} to obtain the averaged and exceptional-set
estimates. We then prove \cref{thm:main_eigenvalue_splitting} by a
two-state reduction. Finally, \cref{appendix:tricomy} supplies the
analytic continuation and bounds for Tricomi's function used in the
argument.

\paragraph{Acknowledgements.}
CF was supported in part by NSF grant DMS-1700180. JS was supported in part by NSF grant DMS-2510207. MIW was supported in part by NSF grants DMS-1908657, DMS-1937254, DMS-2510769  and Simons Foundation Math + X Investigator Award \# 376319 (MIW). Part of this research was carried out during the 2023-24 academic year, when MIW was a Visiting Member in the School of Mathematics - Institute of Advanced Study, Princeton, supported by the Charles Simonyi Endowment, and a Visiting Fellow in the Department of Mathematics at Princeton University.

\section{Proof of \cref{thm:main_hopping}}
We first control the positive part of the logarithmic remainder after
factoring out the exact Gaussian $\ee^{-b\lambda d_1^2}$. The following
averaged estimate yields the exceptional-set bound in
\cref{thm:main_hopping}.
\begin{thm} \label{thm:main_avg_lower_bound_hopping}
    For each fixed $b\in\br{0,b_\star}$, $\varepsilon>0$ and
    $\eta\in\br{0,1}$, there exist $D,C>0$ such that for all $d_1>D$,
    \eql{\label{eq:averaged_gaussian_remainder}
        \frac{1}{\varepsilon d_1}
        \int_{d_1}^{\br{1+\varepsilon}d_1}
        \br{-\log\abs{\rho\br{s,b}}-b\lambda s^2}_+\,\dif s
        \leq C d_1^{1+\eta}.
    }
    Here $t_+:=\max\br{t,0}$.
\end{thm}

In particular, taking $\varepsilon=1/5$ and absorbing the subquadratic
remainder gives, for all sufficiently large $d_1$,
\eq{
    \frac{5}{d_1}\int_{d_1}^{6d_1/5}
    -\log\abs{\rho\br{s,b}}\,\dif s
    \leq \frac32 b\lambda d_1^2.
}
 
To establish these estimates, we use the following lemma.
 \begin{lem}[Complex Analysis Lemma] \label{lem:complex_anal}
    Let \eq{\Gamma_{\alpha}:= \{z\in \CC: -\alpha < \text{Arg}(z) < \alpha\}, ~ 0< \alpha \leq \frac{\pi}{2}.} Suppose $F_1, F_2:\Gamma_\alpha \to \CC$ are holomorphic, $F_2$ is nowhere zero, and
    \begin{itemize} 
        \item there exists  $\beta >0$ such that $\abs{F_1(1)} > \ee^{-\beta}$ and 
        \item  for all $z \in \Gamma_\alpha, \abs{F_1(z)} \leq \abs{F_2(z)}. $
    \end{itemize}
    Then, for every $\varepsilon>0$ and all sufficiently large $R$,
    \eql{
        \frac{1}{\varepsilon R}\int_{t=R}^{(1+\varepsilon )R} -\log\abs{F_1(t)}dt \leq C_{\varepsilon,\alpha} \left(\beta +\log \abs{F_2(1)}\right) R^{\pi/(2\alpha)} + \frac{1}{\varepsilon R}\int_{t=R}^{(1+\varepsilon )R} -\log\abs{F_2(t)}dt.
    }
\end{lem}
The proof follows from Blaschke factorization applied to $F_1/F_2$.
The case $\varepsilon=1$ is presented in
\cite[Corollary 3.3]{FSW_25_lower_bound}. The general case follows by
interval containment and a finite dyadic covering, since
$-\log\abs{F_1/F_2}\geq0$.

For what follows, put $B=b\lambda$ and
$q(z)=\ee^{Bz^2}\rho(z)$. We apply the lemma to this Gaussian quotient,
after an affine change of variable, with $\alpha=\pi/2-\theta$,
where $\theta\in(0,\pi/4)$ is fixed. Write
$p_\theta=\pi/(\pi-2\theta)\in(1,2)$; this exponent can be taken
arbitrarily close to $1$.
We must establish holomorphicity, a global upper bound for $q$ on a
translated sector, and a lower bound at a single point.

 \subsection[Holomorphicity of rho(d1)]{Holomorphicity of $d_1\mapsto\rho(d_1).$} \label{subsec:holomorphicity_of_rho}
To prove the theorem, we first rewrite the hopping coefficient into the form we can easily manipulate. 
\begin{claim}
    For $d_1>a$, put
    \eq{
    \nu=\frac12-\frac{e_0}{2B},\qquad
    Q_{d_1}(x,y)=(x_1-y_1-2d_1)^2+(x_2-y_2)^2.
    }
    Then the hopping coefficient can be written as
\eql{\label{eq:def_hopping_coef}
 \rho(d_1)&=-\frac{\lambda^4\Gamma(\nu)}{4\pi}
 \int_{\RR^2}\int_{\RR^2}
 \overline{v(x)\varphi(x)}v(y)\varphi(y)
 \ee^{-\ii\frac B2x\wedge y}\ee^{\ii B(x_2+y_2)d_1}\nonumber\\
 &\qquad\times U\left(\nu,1,\frac B2Q_{d_1}(x,y)\right)
 \ee^{-\frac B4Q_{d_1}(x,y)}\,dy\,dx,
}
where $U$ is Tricomi's confluent hypergeometric function,
\eq{
 U(\nu,1,z)=\frac1{\Gamma(\nu)}\int_0^\infty
 \ee^{-tz}t^{\nu-1}(1+t)^{-\nu}\,dt
 \qquad(\nu>0,\ \operatorname{Re}z>0),
}
continued analytically as in \cref{appendix:tricomy}.
\end{claim}
\begin{proof}
Let $\varphi$ be the normalized ground state of the single-well Hamiltonian
$h$, with eigenvalue $e_0<0$. The Landau Hamiltonian is
$H^{\mathrm{Landau}}_\lambda=(P-\frac12BX^\perp)^2$.
With $x\wedge y=x_1y_2-x_2y_1$, its resolvent at $e_0$ has kernel
\eql{\label{eq:landau_resolvent_kernel}
 G_{\mathrm{Landau}}(x,y)
 =\frac{\Gamma(\nu)}{4\pi}\ee^{-\ii\frac B2x\wedge y}
 U\left(\nu,1,\frac B2\norm{x-y}^2\right)
 \ee^{-\frac B4\norm{x-y}^2};
}
see \cite[Lemma 5.1]{Korotyaev2003ATF}.
Since $-\lambda^2v_R\varphi_R=(H^{\mathrm{Landau}}_\lambda-e_0)\varphi_R$,
\eq{
 \varphi_R(x)=-\lambda^2\int_{\RR^2}
 G_{\mathrm{Landau}}(x,y)v_R(y)\varphi_R(y)\,dy.
}
We use magnetic translations
\eq{
 (\hat R^u f)(x)=\ee^{-\ii\frac B2x\wedge u}f(x-u),
 \qquad \varphi_L=\hat R^{-d}\varphi,\quad
 \varphi_R=\hat R^d\varphi,
}
as in \cite{Zak1964,SIAM2022_FSW}. Substitution into the hopping coefficient
and translation of the integration variables give
\eq{
 \rho(d_1)&:=\lambda^2\int_{\RR^2}\overline{v_L(x)\varphi_L(x)}\varphi_R(x)\,dx\\
 &=-\lambda^4\iint
 \overline{v(x)\varphi(x)}v(y)\varphi(y)
 \ee^{\ii\frac B2(x_2+y_2)d_1}
 G_{\mathrm{Landau}}(x-d,y+d)\,dy\,dx.
}
Finally, $(x-d)\wedge(y+d)=x\wedge y-d_1(x_2+y_2)$,
which yields \eqref{eq:def_hopping_coef}.
\end{proof}

The integral representation yields the following holomorphic continuation.
\begin{prop}
The function $d_1\mapsto\rho(d_1)$ extends holomorphically to the
half-plane $\HH_a:=\{z\in\CC:\Re{z}>a\}$, where
$\operatorname{supp}(v)\subset B_a$.
\end{prop}
\begin{proof}
    We extend the integrand in the most natural way. Because $v$ is compactly supported, making the integral finite, it is enough to show that the integrand is holomorphic in $d_1$. Observe that each factor of the integrand except the Tricomi's confluent hypergeometric function is explicitly entire. 
    
 It is well known that $U(\nu,1,z)$ is holomorphic for $\nu \in \CC$ and $z \in \CC - (-\infty,0]$ (see \cref{appendix:tricomy}). Therefore, as long as $\frac{b\lambda}{2}((x_1-y_1-2d_1)^2+(x_2-y_2)^2) \notin (-\infty,0]$, we have the holomorphicity. Indeed, for the sake of contradiction, suppose $((x_1-y_1-2d_1)^2+(x_2-y_2)^2) \leq 0.$ Then its imaginary part must be zero which implies that \eq{-4(x_1-y_1-2\text{Re}(d_1))\text{Im}(d_1) = 0} which in turn means either $x_1-y_1-2\text{Re}(d_1)=0$ or $\text{Im}(d_1) = 0$. We know that the first option is not possible because $x_1-y_1 \leq 2a$, yet by construction $\text{Re}(d_1) > a$ and thus $2\text{Re}(d_1) > 2a\geq x_1-y_1.$ If $\text{Im}(d_1) = 0,$ then $\text{Re}((x_1-y_1-2d_1)^2+(x_2-y_2)^2) = (x_1-y_1-2\text{Re}(d_1))^2 + (x_2-y_2)^2 >0.$ Therefore, we have that $\frac{b\lambda}{2}((x_1-y_1-2d_1)^2+(x_2-y_2)^2) \notin (-\infty,0]$ for all choice of $d_1 \in \HH_a, $ and $x,y \in B_a.$

Therefore, the integrand is holomorphic for all $x,y \in B_a$ and is locally uniformly dominated on $\HH_a$. Hence $\rho(d_1)$, and therefore $q(d_1)$, is holomorphic on $\HH_a$.
\end{proof}

With this, we have established holomorphicity of $\rho(d_1).$ Next, we study the upper bound of $\rho(d_1).$

\subsection[Global Upper Bound on the Gaussian Quotient]{Global Upper Bound on the Gaussian Quotient}\label{sec:global_upper_bound_rho}
We retain the exact Gaussian and estimate
$q(z)=\ee^{Bz^2}\rho(z)$. Fix $D_0>a$ and use the translated sector
\eq{
 \Gamma_{\alpha,D_0}:=D_0+\Gamma_\alpha,\qquad \alpha=\pi/2-\theta.
}
\begin{prop}\label{lem:global_ub}
For each fixed $b\in(0,b_\star)$ and $\theta\in(0,\pi/4)$ there exist
$M\geq1$ and $c>0$ such that
\eql{\label{eq:quotient_ub}
 \abs{q(z)}\leq M \ee^{c\operatorname{Re}z}
 \qquad(z\in\Gamma_{\alpha,D_0}).
}
\end{prop}
\begin{proof}
The Gaussian in \eqref{eq:def_hopping_coef} factors as
\eq{
 \ee^{-\frac B4Q_z(x,y)}
 =\ee^{-Bz^2}\ee^{B(x_1-y_1)z}\ee^{-\frac B4\norm{x-y}^2}.
}
For fixed $\nu=\frac12-\frac{e_0}{2B}>0$, put
\eq{
 \zeta_z(x,y)=\frac B2Q_z(x,y).
}
Uniformly for $x,y\in B_a$, we have
\eq{
 \zeta_z(x,y)
 =2Bz^2\left(1+O(\abs{z}^{-1})\right)
 \qquad(\abs{z}\to\infty).
}
Since $z\in\Gamma_{\alpha,D_0}$ and
$\alpha=\pi/2-\theta$, we have $\abs{\arg z}<\pi/2-\theta$.
Consequently, for sufficiently large $\abs{z}$,
\eq{
 \abs{\arg\zeta_z(x,y)}\leq\pi-\theta,
 \qquad
 \abs{\zeta_z(x,y)}\geq B\abs{z}^2,
}
uniformly for $x,y\in B_a$. Applying
\cref{prop:tricomi_sector_bound} with $a=\nu$
and $\eta=\theta$ gives
\eq{
 \abs{U\left(\nu,1,\frac B2Q_z(x,y)\right)}
 \leq
 \frac{\abs{\zeta_z(x,y)}^{-\nu}}{(\sin(\theta/2))^\nu}
 \leq C_{\nu,B,\theta}\abs{z}^{-2\nu}.
}
This gives, for large $\abs{z}$ in our sector,
\eql{\label{eq:rho_UB_estimate}
 \abs{q(z)}&\leq C\abs{z}^{-2\nu}\iint
 \abs{v(x)\varphi(x)v(y)\varphi(y)}
 \abs{\ee^{B(x_1-y_1)z+\ii B(x_2+y_2)z}}\,dy\,dx\nonumber\\
 &\leq C\abs{z}^{-2\nu}\ee^{2aB\abs{z}}.
}
Here the last inequality follows by maximizing the real part of the
linear exponent over $x,y\in B_a$. On $\Gamma_{\alpha,D_0}$,
$\abs{z}\leq(\sin\theta)^{-1}\operatorname{Re}z$. The remaining bounded part of the sector
lies in $\operatorname{Re}z>a$, where the integral is locally uniformly
bounded. Enlarging $M$ gives \eqref{eq:quotient_ub}, with
$c=2aB/\sin\theta$. In particular, $q(z)/(M\ee^{cz})$ is a bounded
holomorphic function on $\Gamma_{\alpha,D_0}$.
\end{proof}

We also record the corresponding decay estimate for the single-well
ground state. There exist $C,c>0$, depending on the fixed $b,\lambda,v$,
such that
\eql{\label{eq:single_well_decay}
 \abs{\varphi(x)}\leq C\ee^{-\frac B4\norm{x}^2+c\norm{x}}
 \qquad(x\in\RR^2).
}
Indeed, the single-well eigenvalue equation gives
\eq{
 \varphi(x)=-\lambda^2\int_{B_a(0)}
 G_{\mathrm{Landau}}(x,y)v(y)\varphi(y)\,\dif y.
}
For $\nu>0$ and $s>0$, the integral representation of $U$ implies
\eq{
 0<U(\nu,1,s)
 \leq \frac{1}{\Gamma(\nu)}
 \int_0^\infty \ee^{-st}t^{\nu-1}\,\dif t
 =s^{-\nu}.
}
Thus, by \cref{eq:landau_resolvent_kernel}, for $\norm{x}\geq a+1$,
\eq{
 \abs{\varphi(x)}
 \leq C\int_{B_a(0)}
 \norm{x-y}^{-2\nu}\ee^{-\frac B4\norm{x-y}^2}
 \abs{v(y)\varphi(y)}\,\dif y
 \leq C\ee^{-\frac B4(\norm{x}-a)^2}.
}
This gives \cref{eq:single_well_decay} for large $\norm{x}$.
Since $\varphi$ is locally bounded, increasing $C$ extends the
estimate to all $x\in\RR^2$.

\subsection{Lower Bound at a Point}

Next, we find a lower bound at a single point.  To do this, we first study the asymptotics of $\rho$ when $b=0$, i.e., the non-magnetic case. Then, by continuity of $\rho$ in $b$, we are guaranteed to have a lower bound up to a factor for small enough $b.$
\begin{lem}
For the nonmagnetic system, as $d_1\to\infty$,
\eq{
 \abs{\rho_{b=0}(d_1)}\sim C_{v,\lambda}
 \frac{\lambda^4\ee^{-2\sqrt{\abs{e_0^{b=0}}}d_1}}{\sqrt{d_1}},
 \qquad C_{v,\lambda}>0.
}
\end{lem}
\begin{proof}
The single-well Hamiltonian is $h_{b=0}=-\Delta+\lambda^2v$.
Put $\kappa=\sqrt{\abs{e_0^{b=0}}}$ and choose its normalized ground state
$\varphi_{b=0}$ strictly positive. The equation
$(-\Delta+\kappa^2)\varphi_{b=0}=-\lambda^2v\varphi_{b=0}$ gives
\eql{\label{eq:ground_state_b=0}
 \varphi_{b=0}(x)=-\lambda^2\int G_{b=0}(x,y)
 v(y)\varphi_{b=0}(y)\,dy,\qquad
 G_{b=0}(x,y)=\frac1{2\pi}K_0(\kappa\norm{x-y}).
}
As $r\to\infty$,
\eq{
 K_0(r)\sim\sqrt{\frac\pi{2r}}\ee^{-r};
}
see \cite[Equation 9.7.2]{abramowitz1972handbook}.
Substitution into the hopping coefficient and translation of the variables
give
\eq{
 \rho_{b=0}(d_1)=-\frac{\lambda^4}{2\pi}\iint
 v(x)v(y)\varphi_{b=0}(x)\varphi_{b=0}(y)
 K_0(\kappa\norm{x-y-2d})\,dy\,dx.
}
Uniformly on $x,y\in B_a$,
$\norm{x-y-2d}=2d_1-(x_1-y_1)+O(d_1^{-1})$.
The Bessel asymptotic therefore gives the claimed formula; its leading
constant is positive because $v\varphi_{b=0}\leq0$ is nonzero.
\end{proof}

This provides the lower bound at one fixed point that we need.
\begin{prop}\label{prop:lowerbound_single_point}
Fix a real $d_1^\star>D_0$. There exists $b_\star>0$ such that
\eql{
 \abs{\rho(d_1^\star,b)}\geq\tfrac12\abs{\rho(d_1^\star,0)}>0
 \qquad(0<b<b_\star).
}
\end{prop}
\begin{proof}
At $b=0$, positivity of the ground state and $v\leq0$, $v\not\equiv0$,
give $\rho(d_1^\star,0)<0$. The assertion follows by continuity at this
fixed point. We fix $d_1^\star$ before choosing $b_\star$.
\end{proof}

\subsection{Applying \cref{lem:complex_anal}.}
Let $\ell=d_1^\star-D_0>0$ and $d_1(z)=D_0+\ell z$, which maps
$\Gamma_\alpha$ onto $\Gamma_{\alpha,D_0}$ and satisfies
$d_1(1)=d_1^\star$. Set
\eq{
 F_1(z)=(q\circ d_1)(z),\qquad
 F_2(z)=M \ee^{c d_1(z)}.
}
By \cref{lem:global_ub,prop:lowerbound_single_point},
$\abs{F_1}\leq\abs{F_2}$ and $F_1(1)\ne0$. Choose $\beta>0$ with
$\abs{F_1(1)}>\ee^{-\beta}$. For real $s>D_0$, put
\eq{
 u(s)=\log M+cs-\log\abs{q(s)}\geq0.
}
\cref{lem:complex_anal}, with $\alpha=\pi/2-\theta$, gives
\eql{\label{eq:lemma_using_last_step}
 \frac1{\varepsilon R}\int_R^{(1+\varepsilon)R}
 u(D_0+\ell t)\,dt\leq C_{\varepsilon,\theta} R^{p_\theta}.
}
Here and below the constants may depend on the fixed $b,\lambda,v,\theta$
and the chosen anchor. To obtain an interval starting at $d_1$, take
$R=(d_1-D_0)/\ell$. For $d_1\geq2D_0$, the image of
$[R,(1+2\varepsilon)R]$ contains $[d_1,(1+\varepsilon)d_1]$.
Using $u\geq0$ and \eqref{eq:lemma_using_last_step} with $2\varepsilon$
therefore gives
\eql{\label{eq:normalized_log_average}
 \frac1{\varepsilon d_1}\int_{d_1}^{(1+\varepsilon)d_1}u(s)\,ds
 \leq C_{\varepsilon,\theta} d_1^{p_\theta}.
}
Since $\log M+cs\geq0$ and $u\br{s}\geq0$, we have
\eq{
 \br{-\log\abs{\rho\br{s}}-Bs^2}_+
 =\br{-\log\abs{q\br{s}}}_+
 \leq u\br{s}.
}
For each fixed $\eta\in\br{0,1}$, choose $\theta>0$ sufficiently small
that $p_\theta\leq1+\eta$. Then \cref{eq:normalized_log_average}
proves \cref{thm:main_avg_lower_bound_hopping}.

Since $-\log\abs{\rho(s)}=Bs^2-\log M-cs+u(s)$, we conclude that
\eql{\label{eq:averaged_gaussian_bound}
 \frac1{\varepsilon d_1}\int_{d_1}^{(1+\varepsilon)d_1}
 -\log\abs{\rho(s)}\,ds
 &\leq B\left(1+\varepsilon+\frac{\varepsilon^2}{3}\right)d_1^2
       +C_{\varepsilon,\theta} d_1^{p_\theta}.
}
For $\varepsilon=1/5$, the quadratic coefficient is $91/75<3/2$.
Since $p_\theta<2$, for each fixed $b>0$ the remainder can thus be
absorbed for sufficiently large $d_1$, giving the weaker averaged bound
stated after \cref{thm:main_avg_lower_bound_hopping}.
By taking $\theta>0$ sufficiently small, the remainder in
\cref{eq:averaged_gaussian_bound} is
$O_{\varepsilon,\eta}(d_1^{1+\eta})$ for every fixed $\eta>0$.

To prove the pointwise statement, fix $\delta\in(0,1)$ and choose
$\theta=\pi\delta/(4+2\delta)$, so that $p_\theta=1+\delta/2$. Define
\eq{
 Z_\delta=\{s>D_0:\abs{\rho(s)}<\ee^{-Bs^2-s^{1+\delta}}\}.
}
On $Z_\delta$, $u(s)\geq s^{1+\delta}$. Applying
\eqref{eq:normalized_log_average} with $\varepsilon=1$ gives
\eq{
 \frac{\abs{Z_\delta\cap[R,2R]}}{R}
 \leq \frac{1}{R^{2+\delta}}\int_R^{2R}u(s)\,ds
 \leq C_{b,\delta}R^{-\delta/2}
 \xrightarrow[R\to\infty]{}0.
}
Summing over dyadic intervals shows that $Z_\delta$ has density zero.
Outside $Z_\delta$, the desired lower bound holds by definition, proving \cref{thm:main_hopping} after increasing $D$ if needed and
restricting $Z_\delta$ to $[D,\infty)$. Together with
\eqref{eq:quotient_ub}, this also gives
\eq{
 \log\abs{\rho(d_1)}=-Bd_1^2+O_\delta(d_1^{1+\delta})
 \qquad(d_1\to\infty,\ d_1\notin Z_\delta).
}

If one retains the weaker threshold, let
\eq{
 \widetilde Z=\{s>D_0:\abs{\rho(s)}<\ee^{-(3/2)Bs^2}\}.
}
For any fixed $\theta\in(0,\pi/4)$, on $\widetilde Z$ we have
$u(s)\geq\frac12Bs^2$, and the same argument gives
$\abs{\widetilde Z\cap[R,2R]}\leq C_{b,\theta}R^{p_\theta-1}$.
Since $p_\theta\downarrow1$ as $\theta\downarrow0$, dyadic summation gives
$\abs{\widetilde Z\cap[D_0,R]}=O_{b,\eta}(R^\eta)$ for every fixed $\eta>0$.

\section{Proof of \cref{thm:main_eigenvalue_splitting}}
\begin{prop}
For each fixed $b\in(0,b_\star)$ and $\delta\in(0,1)$, along
$d_1\to\infty$ outside the density-zero set $Z_\delta$ from \cref{thm:main_hopping},
\eql{
 \abs{\frac{\Delta(d_1)}{2\abs{\rho(d_1)}}-1}=o(1).
}
\end{prop}
\begin{proof}
We follow the two-state reduction used in
\cite[Sections 3--4]{SIAM2022_FSW}, with $b,\lambda$ fixed and
$d_1\to\infty$. We write the off-diagonal matrix element as
$\tau=\ip{\varphi_L}{\br{H-e_0}\varphi_R}=\rho$.

First, the diagonal terms satisfy
\eql{\label{eq:LHL_ineq1}
 \abs{\ip{\varphi_L}{(H-e_0)\varphi_L}}
 &=\lambda^2\abs{\ip{\varphi_L}{v_R\varphi_L}}\nonumber\\
 &\leq C\left(\max_{x\in B_a(0)}\abs{\varphi_L(x+d)}\right)^2,
}
and likewise with $L,R$ exchanged.
Since $\abs{\varphi_L(x+d)}=\abs{\varphi(x+2d)}$,
\cref{eq:single_well_decay} gives, for all sufficiently large $d_1$,
\eql{\label{eq:LHL_ineq2}
 \max_{x\in B_a(0)}\abs{\varphi_L(x+d)}
 \leq C\ee^{-Bd_1^2+c_1d_1}.
}
Thus, writing $r_j=(H-e_0)\varphi_j$ for $j=L,R$, we have
\eq{
 \|r_j\|\leq C \ee^{-Bd_1^2+c_1d_1},\qquad
 \abs{\langle\varphi_j,r_j\rangle}
 \leq C \ee^{-2Bd_1^2+2c_1d_1}.
}
The overlap $\sigma=\langle\varphi_L,\varphi_R\rangle$ tends to zero,
since its modulus is bounded by the overlap of the ordinary translates
of $\abs{\varphi}\in L^2$.

For completeness, we show that the inverse needed in the reduction is bounded uniformly
in $d_1$. Let $P$ be the orthogonal projection onto
$\operatorname{span}\{\varphi_L,\varphi_R\}$ and $Q=1-P$.
All energy expressions below are understood as quadratic forms.

Put
\eq{
 K_B=\left(-\ii\nabla-\frac B2X^\perp\right)^2,
 \qquad h_j=K_B+\lambda^2v_j,\quad j=L,R.
}
The single-well ground state is simple and isolated. By magnetic
translation equivalence, there is a constant $\gamma>0$, independent of
$d_1$, such that
\eq{
 \langle f,(h_j-e_0)f\rangle
 \geq\gamma\left(\|f\|^2-\abs{\langle\varphi_j,f\rangle}^2\right),
 \qquad j=L,R.
}

Choose a localization radius $R>a$, to be fixed below, and assume
$d_1>2R$. Take real, smooth cutoff functions
$\chi_L,\chi_R,\chi_{\mathrm{ext}}$, with values in $[0,1]$, such that
\eq{
 \chi_L^2+\chi_R^2+\chi_{\mathrm{ext}}^2=1,
 \qquad
 \sum_{j\in\{L,R,\mathrm{ext}\}}\norm{\nabla\chi_j}^2
 \leq \frac{C}{R^2}.
}
Here $\chi_L=1$ on $B_R(-d)$ and is supported in $B_{2R}(-d)$,
while $\chi_R=1$ on $B_R(d)$ and is supported in $B_{2R}(d)$.
The constant $C$ can be chosen independently of $R$ and $d_1$.
The two well neighborhoods are disjoint, and
$\chi_{\mathrm{ext}}$ vanishes on the supports of both potentials.

For $\psi\in\operatorname{Ran}Q$ in the quadratic-form domain of $H$,
the magnetic IMS localization formula gives
\eq{
 \langle\psi,(H-e_0)\psi\rangle
 =
 \sum_{j\in\{L,R,\mathrm{ext}\}}
 \langle\chi_j\psi,(H-e_0)\chi_j\psi\rangle
 -
 \int_{\RR^2}
 \left(\sum_{j\in\{L,R,\mathrm{ext}\}}\norm{\nabla\chi_j}^2\right)
 \abs{\psi}^2\,dx.
}
This identity follows from
\eq{
 (-\ii\nabla-A)(\chi_j\psi)
 =\chi_j(-\ii\nabla-A)\psi-\ii(\nabla\chi_j)\psi,
 \qquad A(x)=\frac B2x^\perp,
}
and $\sum_j\chi_j\nabla\chi_j=0$.

Define
\eq{
 \epsilon_R=\|\mathbbm{1}_{\{\norm{x}\geq R\}}\varphi\|_{L^2},
 \qquad \epsilon_R\longrightarrow0\quad(R\to\infty).
}
Since $\psi$ is orthogonal to $\varphi_L$ and $\varphi_R$,
\eq{
 \abs{\langle\varphi_j,\chi_j\psi\rangle}
 =
 \abs{\langle(\chi_j-1)\varphi_j,\psi\rangle}
 \leq \epsilon_R\|\psi\|,
 \qquad j=L,R.
}
The tail bound is independent of $d_1$, because magnetic translations
preserve the modulus of the translated ground state.

On the support of $\chi_j$, the other well potential vanishes.
The single-well gap estimate therefore yields
\eq{
 \sum_{j=L,R}
 \langle\chi_j\psi,(H-e_0)\chi_j\psi\rangle
 \geq
 \gamma\sum_{j=L,R}\|\chi_j\psi\|^2
 -2\gamma\epsilon_R^2\|\psi\|^2.
}
On the support of $\chi_{\mathrm{ext}}$, both potentials vanish.
Since $K_B\geq0$ and $e_0<0$,
\eq{
 \langle\chi_{\mathrm{ext}}\psi,
 (H-e_0)\chi_{\mathrm{ext}}\psi\rangle
 \geq -e_0\|\chi_{\mathrm{ext}}\psi\|^2.
}
Consequently, with $g_0=\min\{\gamma,-e_0\}>0$,
\eq{
 \langle\psi,(H-e_0)\psi\rangle
 \geq
 \left(g_0-2\gamma\epsilon_R^2-\frac{C}{R^2}\right)\|\psi\|^2.
}
Fix $R$ sufficiently large that
$2\gamma\epsilon_R^2+C/R^2\leq g_0/2$, and put $g=g_0/2$.
For all sufficiently large $d_1$, this proves
\eq{
 Q(H-e_0)Q\geq gQ.
}
The compression is understood as the self-adjoint operator on
$\operatorname{Ran}Q$ associated with the restricted quadratic form.
Its resolvent therefore satisfies
\eq{
 \left\|\left[Q(H-e_0-\zeta)Q\right]^{-1}\right\|
 _{\operatorname{Ran}Q\to\operatorname{Ran}Q}
 \leq \frac{2}{g},
 \qquad \zeta\in\CC,\quad \abs{\zeta}\leq\frac g2.
}
All constants and the fixed localization radius may depend on
$b,\lambda,v$, but are independent of sufficiently large $d_1$.

The Schur complement on $\operatorname{Ran}P$, for energies within a fixed
neighborhood of $e_0$, therefore has a correction bounded by
$C(\|r_L\|^2+\|r_R\|^2)$.
Orthonormalizing the two localized states changes their matrix by
$O(\abs{\sigma}\,\abs{\tau}+\ee^{-2Bd_1^2+2c_1d_1})$.
The min--max principle places precisely the lowest two eigenvalues near
$e_0$. Consequently,
\eq{
 \Delta(d_1)=2\abs{\tau(d_1)}
 +O\!\left(\abs{\sigma}\,\abs{\tau(d_1)}+\ee^{-2Bd_1^2+2c_1d_1}\right).
}
Finally, \cref{thm:main_hopping} gives, outside $Z_\delta$,
\eq{
 \frac{\ee^{-2Bd_1^2+2c_1d_1}}{\abs{\rho(d_1)}}
 \leq C \ee^{-Bd_1^2+C_\delta d_1^{1+\delta}+2c_1d_1}\longrightarrow0.
}
Here the limit follows from $\delta<1$. Together with $\sigma\to0$,
this proves the proposition. Increasing $D$
and adjusting $C$ proves \cref{thm:main_eigenvalue_splitting}.
\end{proof}

\appendix

\section{Analytic Continuation of Tricomi's Confluent Hypergeometric Function}\label[appendix]{appendix:tricomy} 
In this section, we study the Tricomi's confluent hypergeometric function, defined by:
\eql{ \label{eq:tricomi_positive_z}
    U(a,1,z) := \frac{1}{\Gamma(a)}\int_0^\infty \ee^{-tz}t^{a-1}(1+t)^{-a}dt.
} where $a >0$ is fixed and $z \in \RR^{>0}.$
 In particular, we study its analytic continuation to $\CC - (-\infty,0]$ as used in \cref{subsec:holomorphicity_of_rho}. 

\begin{prop} Fix $a > 0.$ 
    $U(a, 1,z)$ can be analytically continued to $\CC\setminus(-\infty,0].$
\end{prop}
\begin{proof}
    We start with its formula for positive $z,$ given above in \eqref{eq:tricomi_positive_z}. Observe that we can use the same formula to define $U(a,1,z)$ for $z$ in the right half plane, $\{z \in \CC: \Re{z} >0\}$. Indeed, for $\Re{z} >0$, the integral converges absolutely because \eq{\abs{\ee^{-zt}t^{a -1 }(1+t)^{-a}} = \begin{cases}
        O(t^{a -1}) &\text{ as } t \to 0\\
        O(\ee^{-\Re{z} t}t^{-1}) &\text{ as } t \to \infty.
    \end{cases}}
    Now, we extend $U(a ,1,z)$ to $\CC-(-\infty,0].$
    To do this, we start by defining \eq{\HH_\theta := \{z \in \CC : \Re{z \ee^{\ii\theta}} > 0\} = \left\{z \in \CC : z \neq 0, \Arg{z} \in \left(-\frac{\pi}{2}-\theta, \frac{\pi}{2}-\theta\right)\right\} .}
   Since \eq{\CC-(-\infty, 0 ] = \bigcup_{ \theta \in [-\pi/2,\pi/2]} \HH_\theta,} it is enough for us to extend $U(a,1,z)$ on each $\HH_\theta$ in a consistent manner. That is, if we write $U(a,1,z) =: U_0(a,1,z)$  and $U_\theta(a,1,z)$ to be its extension on $\HH_\theta,$ we show that given $\theta_1\neq \theta_2,$ if $\HH_{\theta_1} \cap \HH_{\theta_2} \neq \emptyset,$ then for all $z \in \HH_{\theta_1} \cap \HH_{\theta_2},$ we have $U_{\theta_1}(a,1,z) = U_{\theta_2}(a,1,z).$ 

   For every $\theta \in [-\pi/2,\pi/2]$, we define \eq{U_\theta(a,1,z) := \lim_{R \to \infty } \frac{1}{\Gamma(a)} \int_0^{R\ee^{\ii\theta}} \ee^{-zt}t^{a-1}(1+t)^{-a} dt .} One can easily check that the integral converges absolutely and that $U_\theta$ is analytic on $\HH_\theta.$ The fact that $U_\theta$ is analytic continuation of $U(a,1,z) $ in  a consistent manner follows easily by elementary Cauchy integral argument on the integrand as a holomorphic function in $t$ for every fixed $z$. Indeed, one can consider an indented sector contour with two edges at angles $\theta_1$ and $\theta_2$ and will see that integral over the two arc parts goes to $0$ uniformly as $\varepsilon\to 0$ and $R \to \infty.$  
 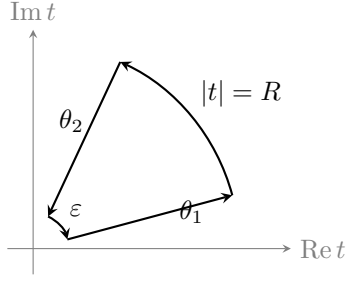
\begin{figure}[ht]
\centering
\begin{tikzpicture}[scale=0.85,>=stealth]
  \draw[->,gray] (-0.4,0)--(4,0) node[right] {$\operatorname{Re}t$};
  \draw[->,gray] (0,-0.4)--(0,3.4) node[above] {$\operatorname{Im}t$};
  \draw[thick,->] (15:0.55)--(15:3.2);
  \draw[thick,->] (15:3.2) arc[start angle=15,end angle=65,radius=3.2];
  \draw[thick,->] (65:3.2)--(65:0.55);
  \draw[thick,->] (65:0.55) arc[start angle=65,end angle=15,radius=0.55];
  \node[right] at (15:2.2) {$\theta_1$};
  \node[left] at (65:2.2) {$\theta_2$};
  \node[above right] at (40:3.2) {$\abs{t}=R$};
  \node[above right] at (40:0.55) {$\varepsilon$};
\end{tikzpicture}
\caption{An indented sector contour, drawn schematically for $0<\theta_1<\theta_2<\pi/2$.}
\end{figure}

The powers in the rotated integral are taken on the branches obtained
from the positive real axis. The small arc contributes $O(\varepsilon^a)$.
For $z$ in a fixed compact subset of $\HH_{\theta_1}\cap \HH_{\theta_2}$,
there exists $c>0$ such that $\Re{ze^{i\theta}}\ge c$
uniformly for all $\theta$ between $\theta_1$ and $\theta_2$.
Thus, on the large arc $t=Re^{i\theta}$, we have
$\abs{e^{-zt}}\le e^{-cR}$, and the large-arc integral tends to zero
exponentially as $R\to\infty$.
Finally, $U(a,1,z)=-\Gamma(a)^{-1}\log z+O(1)$ as
$z\to0$ within a closed subsector of the slit plane, so this continuation
does not extend to an entire function.
\end{proof}
\begin{prop}[A bound on the analytic continuation]
\label{prop:tricomi_sector_bound}
Fix $a>0$ and $\eta\in(0,\pi)$. On the principal branch,
\eq{
 \abs{U(a,1,z)}
 \leq \frac{\abs{z}^{-a}}{(\sin(\eta/2))^a},
 \qquad z\neq0,\quad \abs{\arg z}\leq\pi-\eta.
}
\end{prop}

\begin{proof}
Write $z=\abs{z}\ee^{\ii\vartheta}$ and choose
$\theta=-\vartheta/2$. Then $\abs{\theta}<\pi/2$, and the
rotated-contour representation established above gives
\eq{
 U(a,1,z)
 =
 \frac{\ee^{\ii a\theta}}{\Gamma(a)}
 \int_0^\infty
 \ee^{-zr \ee^{\ii\theta}}r^{a-1}
 (1+r\ee^{\ii\theta})^{-a}\,dr.
}
Since $\abs{\vartheta}\leq\pi-\eta$,
\eq{
 \operatorname{Re}(z\ee^{\ii\theta})
 =\abs{z}\cos(\vartheta/2)
 \geq \abs{z}\sin(\eta/2).
}
Moreover, $\cos\theta\geq0$, so
\eq{
 \abs{1+r\ee^{\ii\theta}}^2
 =1+2r\cos\theta+r^2\geq1.
}
Taking absolute values and evaluating the resulting Gamma integral,
we obtain
\eq{
 \abs{U(a,1,z)}
 \leq
 \frac1{\Gamma(a)}
 \int_0^\infty
 \ee^{-\abs{z}\sin(\eta/2)r}r^{a-1}\,dr
 =
 \frac{\abs{z}^{-a}}{(\sin(\eta/2))^a}.
}
\end{proof}

\newpage
\printbibliography
	\end{document}